\RequirePackage{amsmath}
\documentclass{llncs}
\usepackage{graphicx}
\usepackage{amssymb}
\usepackage{multirow}
\usepackage{color}
\usepackage{xspace}

\usepackage[table]{xcolor}

\usepackage{enumitem}

\usepackage{ifpdf}
\ifpdf    
\usepackage{hyperref}
\else    
\usepackage[hypertex]{hyperref}
\fi

\newcommand{\remove}[1]{}

\newtheorem{inv}[lemma]{Invariant}
\newtheorem{cor}[lemma]{Corollary}

\newcommand{\comment}[1]{}
\newcommand{\junk}[1]{}

\newcommand{\ignore}[1]{}

\newcommand{\paren}[1]{{\mathopen{}\left( #1 \right)\mathclose{}}}

\newcommand{\ceil}[1]{\lceil #1 \rceil}

\newcommand{\ThreeSUM}{\textsf{3SUM}}

\newcommand{\IntegerThreeSUM}{\textsf{Integer\ThreeSUM}}

\newcommand{\Patrascu}{P\v{a}tra\c{s}cu}

\begin{document}
\pagestyle{plain}

\title{Dynamic Set Intersection\thanks{Supported by NSF grants CCF-1217338 and CNS-1318294 and a grant from
the US-Israel Binational Science Foundation.
This research was performed in part at the Center for Massive Data
Algorithmics (MADALGO) at Aarhus University,
which is supported by the Danish National Research Foundation grant DNRF84.}}
\author{
Tsvi Kopelowitz\inst{1} 
and Seth Pettie\inst{1} %
and Ely Porat\inst{2}} 

\institute{University of Michigan \and Bar-Ilan University }


\maketitle

\begin{abstract}
Consider the problem of maintaining a family $F$ of dynamic sets subject to insertions, deletions, and set-intersection reporting queries: given $S,S'\in F$, report every member of
$S\cap S'$ in any order.
We show that in the word RAM model, where $w$ is the word size, given a cap $d$ on the maximum size of any set, we can support set intersection queries in
$O(\frac{d}{w/\log^2 w})$ expected time, and updates in $O(\log w)$ expected time.
Using this algorithm we can list all $t$ triangles of a graph $G=(V,E)$ in $O(m+\frac{m\alpha}{w/\log^2 w} +t)$
expected time, where $m=|E|$ and $\alpha$ is the arboricity of $G$.
This improves a 30-year old triangle enumeration algorithm of Chiba and Nishizeki running in $O(m \alpha)$ time.

We provide an incremental data structure on $F$ that supports intersection {\em witness} queries, where we only need to find {\em one} $e\in S\cap S'$.
Both queries and insertions take $O\paren{\sqrt \frac{N}{w/\log^2 w}}$ expected time, where $N=\sum_{S\in F} |S|$.
Finally, we provide time/space tradeoffs for the fully dynamic set intersection reporting problem.
Using $M$ words of space, each update costs $O(\sqrt {M \log N})$ expected time, each reporting query costs
$O(\frac{N\sqrt{\log N}}{\sqrt M}\sqrt{op+1})$ expected time where $op$ is the size of the output,
and each witness query costs $O(\frac{N\sqrt{\log N}}{\sqrt M} + \log N)$ expected time.

\end{abstract}

\section{Introduction}

In this paper we explore the power of {\em word level parallelism} to speed up algorithms for dynamic set intersection and triangle enumeration.
We assume a $w$-bit word-RAM model, $w>\log n$, with the standard repertoire of unit-time operations on
$w$-bit words: bitwise Boolean operations, left/right shifts, addition, multiplication, comparison, and dereferencing.
Using the modest parallelism intrinsic in this model (sometimes in conjunction with tabulation) it is often possible to
obtain a nearly factor-$w$ (or factor-$\log n$) speedup over traditional algorithms.
The {\em Four Russians} algorithm for boolean matrix multiplication is perhaps the oldest algorithm to use this technique.
Since then it has been applied to computing edit distance~\cite{MasekP80}, regular expression pattern matching~\cite{Myers92},
APSP in dense weighted graphs~\cite{Chan10}, APSP and transitive closure in sparse graphs~\cite{Chan12,Chan08}, and more recently,
to computing the Fr\'echet distance~\cite{BuchinBMM14} and solving 3SUM in subquadratic time~\cite{BaranDP08,GronlundP14}.  Refer to~\cite{Chan13} for more examples.

\paragraph{\textbf{Set Intersection.}} The problem is to represent a (possibly dynamic) family of sets $F$ with total size $N=\sum_{S\in F} |S|$ so that given $S,S'\in F$, one can quickly
determine if $S\cap S'=\emptyset$ (emptiness query) or report some $x\in S\cap S'$ (witness query) or report all members of $S\cap S'$.
Let $d$ be an {\em a priori} bound on the size of any set.
We give a randomized algorithm to preprocess $F$ in $O(N)$ time such
that reporting queries can be answered in $O(d / \frac{w}{\log^2 w} + |S\cap S'|)$ {\em expected} time.
Subsequent insertion and deletion of elements can be handled in $O(1)$ expected time.

We give $O(N)$-space structures for the three types of queries when there is no restriction on the size of sets.
For emptiness queries the expected update and query times are $O(\sqrt{N})$; for witness queries the expected update and query times are $O(\sqrt{N\log N})$;
for reporting queries the expected update time is $O(\sqrt{N\log N})$ and the expected query time is $O(\sqrt{N\log N (1+ |S\cap S'|)})$.
These fully dynamic structures do not benefit from word-level parallelism.
When only insertions are allowed we give another structure that handles
both insertions and emptiness/witness queries in $O(\sqrt{N / \frac{w}{\log^2 w}})$ expected time.\footnote{These data structures
offer a tradeoff between space $M$, query time, and update time.  We restricted our attention to $M=O(N)$ here for simplicity.}

\paragraph{\textbf{\ThreeSUM{} Hardness.}} Data structure lower bounds can be proved unconditionally, or conditionally, based on the {\em conjectured} hardness of some
problem.  One of the most popular conjectures for conditional lower bounds is that the \ThreeSUM{} problem (given $n$ real numbers, determine if any three sum to zero)
cannot be solved in truly subquadratic (expected) time, i.e. $O(n^{2-\Omega(1)})$ time. Even if the inputs are integers in the range $[-n^3,n^3]$ (the \IntegerThreeSUM{} problem), the problem is still conjectured
to be insoluble in truly subquadratic (expected) time. See~\cite{Patrascu10,KPP14a,GronlundP14} and the references therein.

\Patrascu{} in~\cite{Patrascu10} showed that the \IntegerThreeSUM{} problem can be reduced to offline set-intersection, thereby obtaining conditional lower bounds for offline data structures for set-intersection. The parameters of this reduction were tightened by us in~\cite{KPP14a}. Converting a conditional lower bound for the offline version of a problem to a conditional lower bound for the incremental (and hence dynamic) version of the same problem is straightforward, and thus we can prove conditional lower bounds for the incremental (and hence dynamic) set intersection problems. In particular, we are able to show that conditioned on the \IntegerThreeSUM{} conjecture, for the incremental emptiness version either the update or query time must be at least $\Omega(N^{1/2-o(1)})$ time. This is discussed in more detail, including lower bounds for the reporting version, in Appendix~\ref{app:3sum_lb}.

\paragraph{\textbf{Related work.}} Most existing set intersection data structures, e.g., \cite{DLM00,BarbayK02,BY04}, work in the comparison model,
where sets are represented as sorted lists or arrays.  In these data structures the main benchmark is the minimum number of comparisons needed to certify the answer.
Bille, Pagh, and Pagh~\cite{BPP07} also used similar word-packing techniques to evaluate expressions of set intersections and unions.
Their query algorithm finds the intersection of $m$ sets with a total of $n$ elements in $O(n/\frac{w}{\log^2 w} + m\cdot op)$ time, where $op$ is the size of the output.
Cohen and Porat~\cite{CP10} designed a {\em static} $O(N)$-space data structure for answering reporting queries in $O(\sqrt{N(1+|S\cap S'|)})$ time, which is only $O(\sqrt{\log N})$ faster than the data structure presented here.

\paragraph{\textbf{Triangle Enumeration.}}  Itai and Rodeh~\cite{ItaiR78} showed that all $t$ triangles in a graph could be enumerated in $O(m^{3/2})$ time.
Thirty years ago Chiba and Nishizeki~\cite{CN85} generalized~\cite{ItaiR78} to show that $O(m\alpha)$ time suffices, where $\alpha$ is the {\em arboricity} of the graph.
This algorithm has only been improved for dense graphs using fast matrix multiplication.
The recent algorithm of Bj\"orklund, Pagh, Williams, and Zwick~\cite{BjorklundPWZ14} shows that when the matrix multiplication exponent $\omega =2$,
triangle enumeration takes $\tilde{O}(\min\{n^2 + nt^{2/3}, m^{4/3} + mt^{1/3}\})$ time.  (The actual running time is expressed in terms of $\omega$.)
We give the first asymptotic improvement to Chiba and Nishizeki's algorithm for graphs that are too sparse to benefit from fast matrix multiplication.
Using our set intersection data structure, we can enumerate $t$ triangles in $O(m + m\alpha/\frac{w}{\log^2 w} + t)$ expected time.\\

For simplicity we have stated all bounds in terms of an arbitrary word size
$w$.  When $w=O(\log n)$ the $w/\log^2 w$ factor becomes $\log n/\log\log n$.

\paragraph{\textbf{Overview of the paper.}}
The paper is structured as follows. In Section~\ref{sect:set_packing} we discuss a packing algorithm for (dynamic) set intersection, and in Section~\ref{sect:triangle_listing} we show how the packing algorithm for set intersection can be used to speed up triangle listing. In Section~\ref{sec:emptiness} we present our data structure for emptiness queries on a fully dynamic family of sets, with time/space tradeoffs. In Section~\ref{sect:packed_witnesses} we combine the packing algorithm for set intersection with the emptiness query data structure to obtain a packed data structure for set intersection witness queries on an incremental family of sets. In Section~\ref{sect:fully_dyn_set_intersection} we present non-packed data structures for emptiness, witness, and reporting set intersection queries on a fully dynamic family of sets, with time/space tradeoffs. Finally, we discuss conditional lower bounds based on the \ThreeSUM{} conjecture for dynamic versions of the set intersection problem in the Appendix.

\section{Packing Sets}\label{sect:set_packing}

\begin{theorem}\label{theorem:packed_set_intersection}
A family of sets $F=\{S_1,\cdots,S_t\}$ with $d > \max_{S\in F} |S|$ can be preprocessed in linear time
to facilitate the following set intersection queries.
Given two $S,S'\in F$, one can find a witness in $S\cap S'$ in $O(\frac{d\log^2 w}{w})$ expected time and list all of the elements of $S\cap S'$ in $O(|S\cap S'|)$ additional expected time.
If $w = O(\log n)$ then the query time is reduced to $O(\frac{d\log\log n}{\log n})$.
Furthermore, updates (insertions/deletions of elements) to sets in $F$ can be performed $O(1)$ expected time,
subject to the constraint that $d > \max_{S\in F} |S|$.
\end{theorem}


\begin{proof}
Every set $S\in F$ is split into $\ell$ buckets $B^S_1,\ldots, B^S_\ell$ where $\ell = \frac{d\log w}{w}$. We pick a function $h$ from a pairwise independent family of hash functions and assign each element $e\in S$ into a bucket $B^S_{h(e)}$. The expected number of elements from a set $S$ in each bucket is $\frac{w}{\log w}$.
We use a second hash function $h'$ from another family of pairwise independent hash functions which reduces the universe size to $w^2$.  An $h'(e)$ value is represented with $2\log w + 1$ bits, the extra {\em control bit} being necessary for certain manipulations described below.
For each $S$ and $i$ we represent $h'(B^S_i)$ as a packed, sorted sequence of $h'$-values.
In expectation each $h'(B^S_i)$ occupies $O(1)$ words, though some buckets may be significantly larger.
Finally, for each bucket $B^S_i$ we maintain a lookup table that translates from $h'(e)$ to $e$.
If there is more than one element that is hashed to $h'(e)$ then all such elements are maintained in the lookup table via a linked list.

Notice that $S\cap S' = \bigcup_{i=1}^{\ell} B^S_i\cap B^{S'}_i$. Thus, we can enumerate $S\cap S'$ by enumerating the intersections of all $B^S_i\cap B^{S'}_i$.  Fix one such $i$.
We first merge the packed sorted lists $h'(B^S_i)$ and $h'(B^{S'}_i)$.  Albers and Hagerup~\cite{AH97} showed that two words of sorted numbers (separated by control bits) can be merged
using Batcher's algorithm in $O(\log w)$ time.  Using this as a primitive we can merge the sorted lists $h'(B^S_i)$ and $h'(B^{S'}_i)$ in time
$O(|B^S_i| + |B^{S'}_i| / (w/\log^2 w))$.  Let $C$ be the resulting list, with control bits set to 0.
Our task is now to enumerate all numbers that appear twice (necessarily consecutively) in $C$.
Let $C'$ be $C$ with control bits set to 1.  We shift $C$ one field to the right ($2\log w + 1$ bit positions) and subtract it from $C'$.\footnote{The control bits stop carries from crossing field boundaries.}
Let $C''$ be the resulting list, with all control bits reset to 0.  A field is zero in $C''$ iff it and its predecessor were identical, so the problem now is to enumerate zero fields.
By repeated halving, we can distill each field to a single bit (0 for zero, 1 for non-zero) in $O(\log\log w)$ time and then take the complement of these bits (1 for zero, 0 for non-zero).
We have now reduced the problem to reading off all the 1s in a $w$-bit word, which can be done in $O(1)$ time per 1 using the most-significant-bit algorithm of~\cite{FW93}.\footnote{This algorithm uses multiplication.  Without unit-time multiplication~\cite{BrodnikMM97}  one can read off the 1s in $O(\log\log w)$ time per 1.
If $w = O(\log n)$ then the instruction set is not as relevant since we can build $o(n)$-size tables to calculate most significant bits and other useful functions.}
For each repeated $h'$-value we lookup all elements in $B^S_i$ and $B^{S'}_i$ with that value and report any occurring in both sets.  Every unit of time spent in this step
corresponds to an element in the intersection or a false positive.

The cost of intersecting buckets $B^S_i$ and $B^{S'}_i$
is
\[
O\paren{1+\paren{\ceil{\frac{|B^S_i|}{w/\log w}}+\ceil{\frac{|B^{S'}_i|}{w/\log w}}}\log w + |B^S_i\cap B^{S'}_i| + f_i},
\]
where $f_i$ is the number of false positives.
The expected value of $f_i$ is $o(1)$ since the expected sizes of $B^S_i$ and $B^{S'}_i$ are $w/\log w$ and for $e\in B^S_i, e'\in B^{S'}_i$, $\Pr(h'(e) = h'(e')) = 1/w^2$.
Thus, the expected runtime for a query is
\begin{align*}
&\sum_{i=1}^{\ell} O\paren{1+\paren{\ceil{\frac{|B^S_i|}{w/\log w}}+\ceil{\frac{|B^{S'}_i|}{w/\log w}}}\log w + |B^S_i\cap B^{S'}_i| + f_i}\\
&= O(\ell \log w + |S\cap S'|)
\; = O\paren{\frac{d\log^2 w}{w} + |S\cap S'|}.
\end{align*}

It is straightforward to implement insertions and deletions in $O(1)$ time in expectation.
Suppose we must insert $e$ into $S$.  Once we calculate $i = h(e)$ and $h'(e)$ we need to insert $h'(e)$
into the packed sorted list representing $h'(B^S_i)$.  Suppose that $h'(B^S_i)$ fits in one word; let it be $D$, with all control bits set to 1.\footnote{If $h'(B^S_i)$ is larger we apply this procedure to each word of the list $h'(B^S_i)$.  It occupies $O(1)$ words in expectation.}
With a single multiplication we form a word $D'$ whose fields each contain $h'(e)$ and whose control bits are zero.
If we subtract $D'$ from $D$ and mask everything but the control bits, the most significant bit identifies the location of the successor of $h'(e)$ in $h'(B^S_i)$.   We can then insert $h'(e)$ into the sorted list in $D$ with $O(1)$ masks and shifts.
The procedure for deleting an element in $O(1)$ time follows the same lines.
\qed \end{proof}

\section{A Faster Triangle Enumeration Algorithm}\label{sect:triangle_listing}

\begin{theorem}\label{theorem:three_corners}
Given an undirected graph $G=(V,E)$ with $m=|E|$ edges and arboricity $\alpha$, all $t$ triangles can be enumerated in
$O(m+\frac{m\alpha}{w/\log^2 w}+t)$ expected time
or in $O\paren{m + \frac{m\alpha}{\log n/\log\log n} + t}$ expected time if $w = O(\log n)$.
\end{theorem}

\begin{proof}
We will make use of the data structure in Theorem~\ref{theorem:packed_set_intersection}. To do this we first find an acyclic orientation of $E$ in which the out-degree of any vertex is
$O(\alpha)$.
Such an orientation can be found in linear time using the peeling algorithm of Chiba and Nishizeki~\cite{CN85}.
Define $\Gamma^+(u) = \{v \;|\; (u,v)\}$ to be the set of out-neighbors of $u$ according to this orientation.
Begin by preprocessing the family $F = \{\Gamma^+(u) \:|\: u\in V\}$, where all sets have size $O(\alpha)$.
For each edge $(u,v)$, enumerate all elements in the intersection $\Gamma^+(u) \cap \Gamma^+(v)$.
For each vertex $w$ in the intersection output the triangle $\{u,v,w\}$.
Since the orientation is acyclic, every triangle is output exactly once.
There are $m$ set intersection queries, each taking $O(1 + \alpha / \max\{\frac{w}{\log^2 w}, \frac{\log n}{\log\log n}\})$ time, aside from the cost
of reporting the output, which is $O(1)$ per triangle.
\qed \end{proof}

\section{Dynamic Emptiness Queries with Time/Space Tradeoff}\label{sec:emptiness}

\begin{theorem}\label{theorem:emptiness_structure}
There exists an algorithm that maintains a family $F$ of dynamic sets using $O(M)$ space where each update costs $O(\sqrt {M})$ expected time, and each emptiness query costs $O(\frac{N}{\sqrt{M}})$ expected time.\end{theorem}

\begin{proof}
Each set $S\in F$ maintains its elements in a lookup table using a perfect dynamic hash function. So the cost of inserting a new element into $S$, deleting an element from $S$, or determining whether some element $x$ is in $S$ is expected $O(1)$ time. Let $N=\sum_{S\in F} |S|$. We make the standard assumption that $N$ is always at least $N'/2$ and at most $2N'$ for some natural number $N'$. Standard rebuilding de-amortization techniques are used if this is not the case.

\paragraph{The Structure.} We say a set $S$ is \textit{large} if at some point $|S|>2N'/\sqrt{M}$, and since the last time $S$ was at least that large, its size was never less than $N'/\sqrt{M}$. If $S$ is not large, and its size is at least $N'/\sqrt{M}$ then we say it is \textit{medium}. If $S$ is neither large nor medium then it is \textit{small}. Notice that the size of a small set is less than $N'/\sqrt{M}=O(N/\sqrt{M})$. Let $L\subseteq F$ be the sub-family of large and medium sets, and let $\ell = |L|$. Notice that $\ell\leq {\sqrt{M}}$. For each set $S\in L$ we maintain a unique integer $1\leq i_S \leq \ell$, and an \textit{intersection-size} dynamic look-up table $T_S$ of size $\ell$ such that for a large set $S'$ we have $T_S[i_{S'}]= |S\cap S'|$. Adding and deleting entries from the table takes expected constant time using hashing. Due to the nature of our algorithm we cannot guarantee that all of the intersection-size tables will always be fully updated. However, we will guarantee the following invariant.

\begin{inv}\label{inv:indicators}
For every two large sets $S$ and $S'$, $T_S[i_{S'}]$ and $T_{S'}[i_{S}]$ are correctly maintained.
\end{inv}

\paragraph{Query.}
For two sets $S,S'\in F$ where either $S$ or $S'$ is not large, say $S$, we determine if they intersect by scanning the elements in $S$ and using the lookup table for $S'$. The time cost is $O(|S|)=O(N'/\sqrt{M})$. If both sets are large, then we examine $T_S[i_{S'}]$ which determines the size of the intersection (by Invariant~\ref{inv:indicators}) and decide accordingly if it is empty or not. This takes $O(1)$ time.

\paragraph{Insertions.}
When inserting a new element $x$ into $S$, we first update the lookup table of $S$ to include $x$. Next, if $S$ was small and remained small then no additional work is done. Otherwise, for each $S'\in L$ we must update the size of $S\cap S'$ in the appropriate intersection-size tables. This is done directly in $O(\sqrt {M})$ time by determining whether $x$ is in $S'$, for each $S'$, via the lookup tables. We briefly recall, as mentioned above, that it is possible that some of the intersection-size tables will not be fully updated, and so incrementing the size of an intersection is only helpful if the intersection size was correctly maintained before.
Nevertheless, as explained soon, Invariant~\ref{inv:indicators} will be guaranteed to hold, which suffices for the correctness of the algorithm since the intersection-size tables are only used when intersecting two large sets.

The more challenging case is when $S$ becomes medium. If this happens we would like to increase $\ell$ by 1, assign $i_S$ to be the new $\ell$,
allocate and initialize $T_S$ in $O(\sqrt{M})$ time, and for each $S'\in L$ we compute $|S\cap S'|$ and insert the answer into $T_S[i_{S'}]$ and $T_{S'}[i_S]$.
This entire process is dominated by the the task of computing $|S\cap S'|$ for each $S' \in L$, taking a total of $O(\sum_{S'\in L}|S|)$ time, which could be as large as $O(N)$ and is too costly. However, this work can be spread over the next
$N'/\sqrt {M}$ insertions made into $S$ until $S$ becomes large. This is done as follows. When $S$ becomes medium we create a list $L_S$ of all of the large and medium sets at this time (without their elements). This takes $O(\sqrt {M})$ time.
Next, for every insertion into $S$ we compute the values of $O(M/N')$ locations in $T_S$ by computing the intersection size of $S$ and each
of $O(M/N')$ sets from $L_S$ in $O(\frac{M}{N'} \cdot \frac{N}{\sqrt{M}}) = O(\sqrt M)$ time. For each such set $S'$ we also update $T_{S'}[i_S]$.
By the time $S$ becomes large we will have correctly computed the values in $T_S$ for all $O(\sqrt M)$ of the sets in $L_S$, and for every set $S'\in L_S$ we will have correctly computed $T_{S'}[i_S]$. It is possible that between the time $S$ became medium to the time $S$ became large, there were other sets such as $S'$ which became medium and perhaps even large, but $S'\not\in L_S$. Notice that in such a case $S\in L_{S'}$ and so it is guaranteed that by the time both $S$ and $S'$ are large, the indicators $T_S[i_{S'}]$ and $T_{S'}[i_{S}]$ are correctly updated, thereby guaranteeing that Invariant~\ref{inv:indicators} holds. Thus the total cost of performing an insertion is $O(\sqrt M)$ expected time.

\paragraph{Deletions.}
When deleting an element $x$ from $S$, we first update the lookup table of $S$ to remove $x$ in $O(1)$ expected time. If $S$ was small and remained small then no additional work is done. If $S$ was in $L$ then we scan all of the $S'\in L$ and check if $x$ is in $S'$ in order to update the appropriate locations in the intersection-size tables. This takes $O(\sqrt {M})$ time.

If $S$ was medium and now became small, we need to decrease $\ell$ by 1, remove the assignment to $i_S$ to be the new $\ell$, delete $T_S$, and for each $S'\in L$ we need to remove $T_{S'}[i_S]$. In addition, in order to accommodate the update process of medium sized sets, for each medium set $S'$ we must remove $S$ from $L_{S'}$ if it was in there.
\qed \end{proof}

\begin{cor}\label{cor:emptiness_structure_linear}
There exists an algorithm that maintains a family $F$ of dynamic sets using $O(N)$ space where each update costs $O(\sqrt {N})$ expected time, and each emptiness query costs $O(\sqrt{N})$ expected time.
\end{cor}

\section{Incremental Witness Queries}\label{sect:packed_witnesses}

\begin{theorem}\label{theorem:combine_set_intersection}
Suppose there exists an algorithm $A$ that maintains a family $F$ of incremental sets, each of size at most $d$,
such that set intersection witness queries can be answered in $O(\frac{d}{\tau_q})$ expected time and inserts can be performed in $O(\tau_u)$ expected time.
Then there exists an algorithm to maintain a family $F$ of incremental sets---with no upper bound on set sizes---that uses
$O(N)$ space and performs insertions and witness queries in $O(\sqrt{N'/\tau_q} )$ expected time, where $N = \sum_{S\in F} |S|$.
\end{theorem}

\begin{proof}
We make the standard assumption that $N$ is always at least $N'/2$ and at most $2N'$ for some natural number $N'$. Standard rebuilding de-amortization techniques are used if this is not the case.
In our context, we say that a set is large if its size is at least $\sqrt{N'\tau_q}$, and is medium if its size is between $\sqrt{N'/\tau_q}$ and $\sqrt{N'\tau_q}$. Each medium and large set $S$ maintains a \emph{stash} of the at most $\sqrt{N'\tau_q}$ last elements that were inserted into $S$ (these elements are part of $S$). This stash is the entire set $S$ if $S$ is medium. If $S$ is large then the rest of $S$ (the elements not in the stash) is called the \emph{primary} set of $S$. Stashes are maintained using algorithm $A$ with $d=\sqrt{N'\tau_q}$. Thus, answering intersection queries between two medium sets takes $O(\sqrt{N'/\tau_q} )$ expected time.

We maintain for each medium and large set $S$ a witness table $P_S$ such that for any large set $S'$ we have that $P_S[i_{S'}]$ is either an element (witness) in the intersection of $S$ and the primary set of $S'$, or null if no such element exists. This works in the incremental setting as once a witness is established it never changes. Since there are at most $\sqrt{N'/\tau_q}$ large sets and at most $\sqrt{N'\tau_q}$ medium sets, the space usage is $O(N')$. If a query is between $S_1$ and $S_2$ and $S_1$ is large, then: (1) if $S_2$ is small we lookup each element in $S_2$ to see if it is in $S_1$, (2) if $S_2$ is medium or large then we use the witness tables to see if there is a witness of an intersection between $S_2$ and the primary set of $S_1$ or between $S_1$ and the primary set of $S_2$, and if there is no such witness then we use algorithm $A$ to intersect the stashes of $S_2$ and $S_1$. In any case, the cost of a query is $O(\sqrt{N'/\tau_q})$ expected time. The details for maintaining these tables are similar to the details of maintaining the intersection-size array tables from Section~\ref{sec:emptiness}.

\paragraph{Insertion.} When inserting an element $x$ into $S$, if $S$ is small then we do nothing. If $S$ is medium then we add $x$ to the stash of $S$ in algorithm $A$. If $S$ is large then we add $x$ to the stash of $S$ and verify for every other large set if $x$ is in that set, updating the witness table accordingly. If $S$ became medium then we add it to the structure of algorithm $A$. Since the size of $S$ is $O(\sqrt{N'/\tau_q})$ this takes $O(\sqrt{N'/\tau_q})$ expected time. Furthermore, when $S$ becomes medium the table $P_S$ needs to be prepared. To do this, between the time $S$ is of size $\sqrt{N'/2\tau_q}$ and the time $S$ is of size $\sqrt{N'/\tau_q}$, the table $P_S$ is inclemently constructed.
If $S$ became large then we now allow its primary set to be nonempty, and must also update the witness tables.
The changes to witness tables in this case is treated using the same techniques as in Theorem~\ref{theorem:emptiness_structure}, and so we omit their description. This will cost $O(\sqrt{N'/\tau_q}+\tau_u)$ expected time.

Finally, for a large set $S$, once its stash reaches size $\sqrt {N' \tau_q}$ we dump the stash into the primary set of $S$, thereby emptying the stash. We describe an amortized algorithm for this process, which is deamortized using a standard lazy approach. To combine the primary set and the stash we only need to update the witness tables for set intersection witnesses between medium sets and the new primary set of $S$ as it is possible that a witness was only in the stash. To do this, we directly scan all of the medium sets and check if a new witness can be obtained from the stash. The number of medium sets is $O(\sqrt {N'\tau_q})$ and the cost of each intersection will be $O(\sqrt{N'/\tau_q})$ for a total of $O(N')$ time. Since this operation only happens after $\Omega(\sqrt {N'\tau_q})$ insertions into $S$ the amortized cost is $O(\sqrt{N'/\tau_q})$ time.

\qed \end{proof}

Combining Theorem~\ref{theorem:packed_set_intersection} with Theorem~\ref{theorem:combine_set_intersection} we obtain the following.
\begin{corollary}
There exists an algorithm in the word-RAM model that maintains a family $F$ of incremental sets using $O(N)$ space where each insertion costs $O(\sqrt {\frac{N}{w/\log^2 w}} +\log w)$ expected time and a witness query costs $O(\sqrt {\frac{N}{w/\log^2 w}} )$ expected time.
\end{corollary}

\section{Fully Dynamic Set Intersection with Witness and Reporting Queries}\label{sect:fully_dyn_set_intersection}

Each element in $\bigcup_{S\in F} S$ is assigned an integer from the range of $[2N']$. When a new element not appearing in $\bigcup_{S\in F} S$ arrives, it is assigned to the smallest available integer, and that integer is used as its key. When keys are deleted (no longer in use), we do not remove their assignment, and instead, we conduct a standard rebuilding technique in order to reassign the elements. Finally, we use a second assignment via a random permutation of the integers in order to uniformly spread the assignments within the range.

\paragraph{The structure.}
Consider the following binary tree $T$ of height $\log N' + 1$ where each vertex $v$ covers some range from $U$, denoted by $[\alpha_v,\beta_v]$, such that the range of the root covers all of $U$, and the left (right) child of $v$ covers the first (second) half of $[\alpha_v,\beta_v]$. A vertex at depth $i$ covers $\frac{2N'}{2^i}$ elements of $U$. For a vertex $v$ let $S^v=S\cap [\alpha_v,\beta_v]$. Let $N_v = \sum_{S\in F} |S^v|$. Let $M_v = \frac{N_v\cdot M}{N'}$. We say a set $S$ is \textit{$v$-large} if at some point $|S^v|>\frac{2N_v}{\sqrt{M_v}}$, and since the last time $S^v$ was at least that large, its size was never less than $\frac{N_v}{\sqrt{M_v}}$.

Each vertex $v\in T$ with children $v_0$ and $v_1$ maintains a structure for emptiness queries as in Theorem~\ref{theorem:emptiness_structure}, using $M_v$ space, on the family $F^v=\{S^v: S\in F\}$. In addition, we add auxiliary data to the intersection-size tables as follows. For sets $S_1,S_2\in F$ the set of all vertices in which $S_1$ and $S_2$ intersect under them defines a connected tree $T'$. This tree has some branching vertices which have 2 children, some non-branching internal vertices with only 1 child, and some leaves. Consider the vertices $v$ in $T$ for which $S_1$ and $S_2$ are $v$-large and define $\hat{T}$ to be the connected component of these vertices that includes the root $r$.  (It may be that $\hat{T}$ does not exist.)
To facilitate a fast traversal of $\hat{T}$ during a query we maintain \textit{shortcut} pointers for every two sets $S_1,S_2\in F$ and for every vertex $v\in T$ such that both $S_1$ and $S_2$ are $v$-large. To this end, we say $v$ is a \textit{branching-$(S_1,S_2)$-vertex} if both $S_1^{v_0}\cap S_2^{v_0} \neq \emptyset$ and $S_1^{v_1}\cap S_2^{v_1} \neq \emptyset$. Consider the path starting from the left (right) child of $v$ and ending at the first descendent $v'$ of $v$ such that:(1) $S_1$ and $S_2$ are relatively large for all of the vertices on the path, (2) $S_1^{v'}\cap S_2^{v'}\neq \emptyset$, and (3) either $v'$ is a \textit{branching-$(S_1,S_2)$-vertex} or one of the sets $S_1$ and $S_2$ is not $v'$-large. The left (right) shortcut pointer of $v$ will point to $v'$. Notice that the shortcut pointers are maintained for every vertex $v$ even if on the path from $r$ to $v$ there are some vertices for which either $S_1$ or $S_2$ are not relatively large, which helps to reduce the update time during insertions/deletions. Also notice that using these pointers it is straightforward to check in $O(1)$ time if $S_1^{v_0}\cap S_2^{v_0}$ and $S_1^{v_1}\cap S_2^{v_1}$ are empty or not.

The space complexity of the structure is as follows. Each vertex $v$ uses $O(M_v)$ words of space which is $O(M N_v/N')$. So the space usage is $\sum_v M_v= O(M\log N)$ words, since in each level of $T$ the sum of all $M_v$ for the vertices in that level is $O(M)$, and there are $O(\log N)$ levels.

\paragraph{Reporting queries.}
For a reporting query on $S_1$ and $S_2$, if $op=0$ then either the emptiness test at the root will conclude in $O(1)$ time, or we spend $O(\frac{N_r}{\sqrt{M_r}}) = O(\frac{N}{\sqrt M})$ time. Otherwise, we recursively examine vertices $v$ in $T$ starting with the root $r$. If both $S_1$ and $S_2$ are $v$-large and $S_1^v\cap S_2^v \neq \emptyset$, then we continue recursively to the vertices pointed to by the appropriate shortcut pointers. If either $S_1$ or $S_2$ is not $v$-large then we wish to output all of the elements in the intersection of $S_1^v$ and $S_2^v$. To do this, we check for each element in the smaller set if it is contained within the larger set using the lookup table which takes $O(\frac{N_v}{\sqrt{M_v}})$ time.

%
For the runtime, as we traverse down $T$ from $r$ using appropriate shortcut pointers, we encounter only two types of vertices. The first type are vertices $v$ for which both $S_1$ and $S_2$ are $v$-large, and the second type are vertices $v$ for which either $S_1$ or $S_2$ is not $v$-large. Each vertex of the first type performs $O(1)$ work, and the number of such vertices is at most the number of vertices of the second type, due to the branching nature of the shortcut pointers. For vertices of the second type, the intersection of $S_1$ and $S_2$ must both be non-empty relative to such vertices and so the $O(\frac{N_v}{\sqrt{M_v}})$ time cost can be charged to at least one element in the output. Denote the vertices of the second type by $v_1,v_2,\ldots,v_t$. Notice that $t\leq op$ as each $v_i$ contains at least one element from the intersection, and that $\sum_i N_{v_i} < 2N'$ since the vertices are not ancestors of each other. We will make use of the following Lemma.
\begin{lemma}\label{lemma:sum_sqrt_bound}
If $\sum_{i=1}^t x_i \leq k$  then $\sum_{i=1}^t \sqrt{x_i} \leq \sqrt{k\cdot t}$.
\end{lemma}
\begin{proof}
Since $\sum_{i=1}^t \sqrt{x_i}$ is maximized whenever all the $x_i$ are equal, we have that $\sum_{i=1}^t \sqrt{x_i} \leq t\sqrt{\frac{k}{t}} = \sqrt{kt}$.
\qed \end{proof}

Therefore, the total time cost is
\begin{align*}
\sum_i \frac{N_{v_i}}{\sqrt{M_{v_i}}} &=\sum_i \frac{N_{v_i}\sqrt{N'}}{\sqrt{M N_{v_i}}}
=\sqrt{\frac{N'}{M}} \sum_i \sqrt{N_{v_i}}
 \leq \sqrt{\frac{N'}{M}}  \sqrt{2N'}\sqrt{t}
\leq O\paren{\frac{N\sqrt{op}}{\sqrt{M}}}.
\end{align*}

\paragraph{Witness queries.} A witness query is answered by traversing down $T$ using shortcut pointers, but instead of recursively looking at both shortcut pointers for each vertex, we only consider one. Thus the total time it takes until we reach a vertex $v$ for which either $S_1$ or $S_2$ is not $v$-large is $O(\log N)$. Next, we use the hash function to find an element in the intersection in $O(\frac{N}{\sqrt M})$ time, for a total of $O(\log N + \frac{N}{\sqrt M})$ time to answer a witness query.

\paragraph{Insertions and Deletions.}
When inserting a new element $x$ into $S_1$, we first locate the leaf $\ell$ of $T$ which covers $x$. Next, we update our structure on the path from $\ell$ to $r$ as follows. Starting from $\ell$, for each vertex $v$ on the path we insert $x$ into $S_1^v$. This incurs a cost of $\sqrt {M_v} $ for updating the emptiness query structure at $v$. If there exists some set $S_2$ such that $|S_1^v\cap S_2^v|$ becomes non-zero, then we may need to update some shortcut pointers on the path from $\ell$ to $r$ relative to $S_1$ and $S_2$. Being that such a set $S_2$ must be large, the number of such sets is at most $\frac{N_v}{\sqrt{M_v}}$.

To analyze the expected running time of an insertion notice that since the elements in the universe are randomly distributed, the expected value of $N_v$ and $M_v$ for a vertex $v$ at depth $i$ are $\frac{N}{2^i}$ and $\frac{M}{2^i}$ respectively. So the number of $v$-large sets is at most $\frac{N_v}{\sqrt{M_v}} = \frac{N}{\sqrt{2^iM}}$. The expected time costs of updating the emptiness structure is at most $\sum_{i=0}^{\log N'}\frac{N}{\sqrt{2^iM}} = O(\frac{N}{\sqrt M})$. The same analysis holds for the shortcut pointer.
The deletion process is exactly the reverse of the insertions process, and also costs $O(\frac{N}{\sqrt M})$ expected time.

The total space usage is $O(M\log N)$.  With a change of variable (substituting $M/\log N$ for $M$ in the construction above),
we can make the space $O(M)$ and obtain the following result.

\begin{theorem}\label{theorem:mem_sqrt_set_intersection}
There exists an algorithm that maintains a family $F$ of dynamic sets using $O(M)$ space where each update costs $O(\sqrt {M \log N})$ expected time, each reporting query costs $O(\frac{N\sqrt{\log N}}{\sqrt M}\sqrt{op+1})$ time, and each witness query costs $O(\frac{N\sqrt{\log N}}{\sqrt M}+ \log N)$ expected time.
\end{theorem}

{\small
\bibliographystyle{splncs}
\bibliography{tsvi}

\begin{thebibliography}{10}

\bibitem{MasekP80}
Masek, W.J., Paterson, M.:
\newblock A faster algorithm computing string edit distances.
\newblock J. Comput. Syst. Sci. \textbf{20}(1) (1980)  18--31

\bibitem{Myers92}
Myers, G.:
\newblock A {F}our {R}ussians algorithm for regular expression pattern
  matching.
\newblock J. ACM \textbf{39}(2) (1992)  432--448

\bibitem{Chan10}
Chan, T.M.:
\newblock More algorithms for all-pairs shortest paths in weighted graphs.
\newblock SIAM J.~Comput. \textbf{39}(5) (2010)  2075--2089

\bibitem{Chan12}
Chan, T.M.:
\newblock All-pairs shortest paths for unweighted undirected graphs in $o(mn)$
  time.
\newblock ACM Transactions on Algorithms \textbf{8}(4) (2012) ~34

\bibitem{Chan08}
Chan, T.M.:
\newblock All-pairs shortest paths with real weights in $o(n^3/\log n)$ time.
\newblock Algorithmica \textbf{50}(2) (2008)  236--243

\bibitem{BuchinBMM14}
Buchin, K., Buchin, M., Meulemans, W., Mulzer, W.:
\newblock Four {S}oviets walk the dog -- with an application to {A}lt's
  conjecture.
\newblock In: Proceedings 25th Annual ACM-SIAM Symposium on Discrete Algorithms
  (SODA). (2014)  1399--1413

\bibitem{BaranDP08}
Baran, I., Demaine, E.D., P\v{a}tra\c{s}cu, M.:
\newblock Subquadratic algorithms for {3SUM}.
\newblock Algorithmica \textbf{50}(4) (2008)  584--596

\bibitem{GronlundP14}
Gr{\o}nlund, A., Pettie, S.:
\newblock Threesomes, degenerates, and love triangles.
\newblock In: Proceedings 55th IEEE Symposium on Foundations of Computer
  Science (FOCS). (2014) Full manuscript available as arXiv:1404.0799.

\bibitem{Chan13}
Chan, T.M.:
\newblock The art of shaving logs.
\newblock In: Proceedings 13th Int'l Symposium on Algorithms and Data
  Structures (WADS). Volume 8037 of Lecture Notes in Computer Science.
\newblock Springer (2013)  231--231

\bibitem{Patrascu10}
P\v{a}tra\c{s}cu, M.:
\newblock Towards polynomial lower bounds for dynamic problems.
\newblock In: Proceedings 42nd ACM Symposium on Theory of Computing (STOC).
  (2010)  603--610

\bibitem{KPP14a}
Kopelowitz, T., Pettie, S., Porat, E.:
\newblock 3sum hardness in (dynamic) data structures.
\newblock CoRR \textbf{abs/1407.6756} (2014)

\bibitem{DLM00}
Demaine, E.D., L{\'o}pez-Ortiz, A., Munro, J.I.:
\newblock Adaptive set intersections, unions, and differences.
\newblock In: Proceedings of the Eleventh Annual ACM-SIAM Symposium on Discrete
  Algorithms. (2000)  743--752

\bibitem{BarbayK02}
Barbay, J., Kenyon, C.:
\newblock Adaptive intersection and t-threshold problems.
\newblock In: Proceedings 13th Annual ACM-SIAM Symposium on Discrete Algorithms
  (SODA). (2002)  390--399

\bibitem{BY04}
Baeza-Yates, R.A.:
\newblock A fast set intersection algorithm for sorted sequences.
\newblock In: Combinatorial Pattern Matching, 15th Annual Symposium, CPM.
  (2004)  400--408

\bibitem{BPP07}
Bille, P., Pagh, A., Pagh, R.:
\newblock Fast evaluation of union-intersection expressions.
\newblock In Tokuyama, T., ed.: ISAAC. Volume 4835 of Lecture Notes in Computer
  Science., Springer (2007)  739--750

\bibitem{CP10}
Cohen, H., Porat, E.:
\newblock Fast set intersection and two-patterns matching.
\newblock Theor. Comput. Sci. \textbf{411}(40-42) (2010)  3795--3800

\bibitem{ItaiR78}
Itai, A., Rodeh, M.:
\newblock Finding a minimum circuit in a graph.
\newblock SIAM J.~Comput. \textbf{7}(4) (1978)  413--423

\bibitem{CN85}
Chiba, N., Nishizeki, T.:
\newblock Arboricity and subgraph listing algorithms.
\newblock SIAM J. Comput. \textbf{14}(1) (1985)  210--223

\bibitem{BjorklundPWZ14}
Bjorklund, A., Pagh, R., Williams, V.V., Zwick, U.:
\newblock Listing triangles.
\newblock In: Automata, Languages, and Programming - 41st International
  Colloquium, {ICALP} 2014, Copenhagen, Denmark, July 8-11, 2014, Proceedings,
  Part {I}. (2014)  223--234

\bibitem{AH97}
Albers, S., Hagerup, T.:
\newblock Improved parallel integer sorting without concurrent writing.
\newblock Inf. Comput. \textbf{136}(1) (1997)  25--51

\bibitem{FW93}
Fredman, M.L., Willard, D.E.:
\newblock Surpassing the information theoretic bound with fusion trees.
\newblock J. Comput. Syst. Sci. \textbf{47}(3) (1993)  424--436

\bibitem{BrodnikMM97}
Brodnik, A., Miltersen, P.B., Munro, J.I.:
\newblock Trans-dichotomous algorithms without multiplication -- some upper and
  lower bounds.
\newblock In: Proceedings 5th Int'l Workshop on Algorithms and Data Structures
  (WADS). Volume 1272 of Lecture Notes in Computer Science.
\newblock Springer Berlin Heidelberg (1997)  426--439

\end{thebibliography}
}

\newpage

\appendix
\section{Conditional Lower Bounds from \ThreeSUM{}}\label{app:3sum_lb}

We first make use of the following Theorem, which was proven by Kopelowitz, Pettie, and Porat~\cite{KPP14a}.

\begin{theorem}[\cite{KPP14a}]\label{thm:improved_reduction_reporting}
For any constants $0\leq \gamma < 1$ and $0<\delta\leq 2$,
let $\mathbb{A}$ be an algorithm for the offline set intersection reporting problem on a family $F$ of sets such that $N = \sum_{S\in F} |S| = \Theta(n^{\frac{3+\delta-\gamma}{2}})$ and there are $\Theta(n^{1+\gamma})$ pairs of sets whose intersection needs to be reported such that the total size of these set intersections of these $t$ pairs is expected to be $O(n^{2-\delta})$. If $\mathbb{A}$ runs in expected $O(n^{2-\Omega(1)})$ time, then \IntegerThreeSUM{} can be solved in expected $O(n^{2-\Omega(1)})$ time.
\end{theorem}

\begin{theorem}\label{thm:dynamic_set_int_report_lb}
{\bf (Set Intersection Reporting Lower Bound)}
For any constants $0\leq \gamma < 1$ and $0<\delta< 1$,
any algorithm for solving the incremental set intersection reporting problem with insertion time of $t_i$ and query time $t_q + t_r\cdot op$ (where $op$ is the size of the output)
must have $N\cdot t_i+N^{\frac{2(1+\gamma)}{3+\delta-\gamma}}t_q + N^{\frac{4-2\delta}{3+\delta-\gamma}}t_r= \Omega(N^{\frac{4}{3+\delta-\gamma}-o(1)})$ unless the \IntegerThreeSUM{} conjecture is false.
\end{theorem}

\begin{proof}
An algorithm for solving the incremental set intersection decision problem can be used to solve \IntegerThreeSUM{} via Theorem~\ref{thm:improved_reduction_reporting} by first inserting all of the $\Theta(n^{\frac{3+\delta-\gamma}{2}})$ elements into their appropriate sets and then performing the $\Theta(n^{1+\gamma})$ queries. Therefore, unless the \IntegerThreeSUM{} conjecture is false, we have $\Theta(n^{\frac{3+\delta-\gamma}{2}} t_i + n^{1+\gamma}t_q + n^{2-\delta}t_r) = \Omega(n^{2-o(1)})$. Substituting $n=N^{\frac{2}{3+\delta-\gamma}}$ completes the proof.
\qed \end{proof}

Let us consider a few points on the lower bound curve of Theorem~\ref{thm:dynamic_set_int_report_lb}. The coefficients of the terms $t_i$, $t_q$, and $t_r$ are equal when $\gamma = \delta = 1/2$, which translates to $t_i+t_q +t_r= \Omega(N^{1/3-o(1)})$. Thus, at least one of the operations must cost roughly $\Omega(N^{1/3})$ time. Furthermore, if $t_q=t_r=O(1)$ then $t_i = \Omega(N^{\frac{4}{3+\delta-\gamma}-1-o(1)})$ so by making $\delta$ as small as possible and $\gamma$ as large as possible we obtain $t_i = \Omega(N^{1-o(1)})$. This matches a trivial algorithm where we explicitly maintain each set intersection. However, if $t_i=t_r=O(1)$ then $t_q = \Omega(N^{\frac{2-2\gamma}{3+\delta-\gamma}-o(1)})$, and so by making $\delta$ as small as possible and setting $\gamma = 0$ we obtain $t_q = \Omega(N^{2/3-o(1)})$. Finally, if $t_i=t_q=O(1)$ then $t_r = \Omega(N^{\frac{2\delta}{3+\delta-\gamma}-o(1)})$ and so making $\gamma$ and $\delta$ as large as possible we obtain $t_r = \Omega(N^{2/3 -o(1)})$.

\begin{theorem}[\cite{KPP14a}]\label{thm:improved_reduction}
For any constant $0< \gamma < 1$
let $\mathbb{A}$ be an algorithm for offline set intersection decision problem on a family $F$ of sets such that $N = \sum_{S\in F} |S| = \Theta(n^{2-\gamma})$, and there are $\Theta(n^{1+\gamma})$ pairs of sets whose disjointness needs to be determined. If $\mathbb{A}$ runs in expected $O(n^{2-\Omega(1)})$ time, then \IntegerThreeSUM{} can be solved in expected $O(n^{2-\Omega(1)})$ time.
\end{theorem}

\begin{theorem}\label{thm:dynamic_set_int_decision_lb}
{\bf (Set Intersection Emptiness Lower Bound)}
Fix $0< \gamma < 1$.  Any algorithm for solving the incremental set intersection emptiness problem with insertion time
$t_i$ and query time of $t_q$ must have $N\cdot t_i+N^{\frac{1+\gamma}{2-\gamma}}t_q = \Omega(N^{\frac{2}{2-\gamma}-o(1)})$ unless the \IntegerThreeSUM{} conjecture is false.
\end{theorem}

\begin{proof}
An algorithm for solving the incremental set intersection decision problem can be used to solve \IntegerThreeSUM{} via Theorem~\ref{thm:improved_reduction} by first inserting all of the $\Theta(n^{2-\gamma})$ elements into their appropriate sets and then performing the $\Theta(n^{1+\gamma})$ queries. Therefore, unless the \IntegerThreeSUM{} conjecture is false, we have
$\Theta(n^{2-\gamma} t_i + n^{1+\gamma}t_q) = \Omega(n^{2-o(1)})$. Substituting $n=N^{\frac{1}{2-\gamma}}$ completes the proof.

\qed \end{proof}

Let us consider a few points on the lower bound curve of Theorem~\ref{thm:dynamic_set_int_decision_lb}. The coefficients of the terms $t_i$ and $t_q$ are equal when $\gamma = 1/2$, which translates to $t_i+t_q = \Omega(N^{1/3-o(1)})$. Thus, at least one of the operations must cost roughly $\Omega(N^{1/3})$ time. Furthermore, if $t_i=O(1)$ then $t_q = \Omega(N^{\frac{1-\gamma}{2-\gamma}-o(1)})$ so by making $\gamma$ as small as possible we obtain $t_q = \Omega(N^{1/2-o(1)})$. Finally, if $t_q=O(1)$ then $t_i = \Omega(N^{\frac{\gamma}{2-\gamma}-o(1)})$ so by making $\gamma$ as large as possible we obtain $t_i = \Omega(N^{1/2-o(1)})$.

\end{document}